\documentclass[sigconf]{aamas} 

\usepackage{balance} 

\setcopyright{ifaamas}
\acmConference[AAMAS '23]{Proc.\@ of the 22nd International Conference
on Autonomous Agents and Multiagent Systems (AAMAS 2023)}{May 29 -- June 2, 2023}
{London, United Kingdom}{A.~Ricci, W.~Yeoh, N.~Agmon, B.~An (eds.)}
\copyrightyear{2023}
\acmYear{2023}
\acmDOI{}
\acmPrice{}
\acmISBN{}

\acmSubmissionID{90}

\title[AAMAS-2023 Formatting Instructions]{Strategic Play By Resource-Bounded Agents in Security Games}

\author{Xinming Liu}
\affiliation{
  \institution{Massachusetts Institute of Technology}
  \city{Cambridge}
  \country{United States of America}}
\email{xinming@mit.edu}

\author{Joseph Y. Halpern}
\affiliation{
  \institution{Cornell University}
  \city{Ithaca}
  \country{United States of America}}
\email{halpern@cs.cornell.edu}

\begin{abstract}
Many studies have shown that humans are ``predictably irrational'': they do not act in a fully rational way, but their deviations from rational behavior are quite systematic. Our goal is to see the extent to which we can explain and justify these deviations as the outcome of rational but resource-bounded agents doing as well as they can, given their limitations. We focus on the well-studied ranger-poacher game, where rangers are trying to protect a number of sites from poaching. We  capture the computational limitations by modeling the poacher and the ranger as probabilistic finite automata (PFAs). We show that, with sufficiently large memory, PFAs learn to play the Nash equilibrium (NE) strategies of the game and achieve the NE utility. However, if we restrict the memory, we get more ``human-like'' behaviors, such as \emph{probability matching} (i.e., visiting sites in proportion to the probability of a rhino being there), and avoiding sites where there was a bad outcome (e.g., the poacher was caught by the ranger), that we also observed in experiments conducted on Amazon Mechanical Turk. Interestingly, we find that adding human-like behaviors such as probability matching and overweighting significant events (like getting caught) actually improves performance, showing that this seemingly irrational behavior can be quite rational. 
\end{abstract}


\newcommand{\BibTeX}{\rm B\kern-.05em{\sc i\kern-.025em b}\kern-.08em\TeX}

\begin{document}


\pagestyle{fancy}
\fancyhead{}


\maketitle 


\section{Introduction} \label{sec:intro}
While standard economic theory assumes that people are rational, many
studies (see, e.g., \cite{Ariely10}) have shown that humans are
irrational in a systematic way. We are interested in the extent to
which these behaviors can be explained by  computational
limitations. Following a tradition that goes back to Rubinstein
\cite{Rub85} and Neyman \cite{Ney85}, we model computationally
bounded agents as probabilistic finite automata (PFAs). Earlier work
(see, e.g., \cite{W02,HPS12,Liu2020MAB}) has shown that optimal finite
automata for certain problems  can exhibit quite ``human-like''
behaviors, such as confirmation and a first-impression bias. Our goal
in this paper is to see the extent to which computational limitations
can explain and justify human behaviors in \emph{security games}
\cite{Tambe12}.     

Specifically, we consider  a (finitely) repeated two-player ranger-poacher game, based on the wildlife poaching game introduced by Kar et al. \cite{Kar15}. At each stage of the repeated game, the poacher tries to catch a rhino at one of $n$ sites, and the ranger tries to prevent the poacher from doing so. We assume that there is a commonly-known probability of a rhino being at any particular site, which does not change over time. We can formulate the stage game (i.e., the game played at each step of the repeated game) as a zero-sum normal-form game, which we show has a unique Nash equilibrium (NE).  An easy backward induction then shows that the unique NE of the finitely repeated game is to play the NE of the stage game repeatedly. We take ``rational behavior'' to be best responding to the opponent (which will mean playing the NE strategy if the opponent is).

There are several well-studied algorithms that lead to NE; we focus on one of them here: \emph{fictitious play} (FP) \cite{Brown51}. In FP, each player keeps track of what the other player has done, and best responds to the mixed strategy where the probability that an action is played is the frequency with which that strategy has been played thus far.\footnote{The term ``fictitious play'' is due to the fact that this procedure could in principle be simulated by each player without the game actually being played.  Note that there is an implicit assumption here that all actions are observed.} Robinson \cite{Robinson51} showed that in  two-player zero-sum game where both players have only a finite number of strategies (which is what we consider here), FP converges to NE strategy and utility. However, the rate of convergence of FP is often slow, and is quite sensitive to which strategy is chosen if there are multiple best responses. Abernethy, Lai, and Wibisono \cite{ALW19} recently showed that for two-player zero-sum games where the payoff matrix is a diagonal matrix, after $t$ steps, players' estimates are within $O(1/\sqrt{t})$ of NE, if they use a specific tie-breaking rule, while Daskalakis and Pan \cite{DP14} showed that with a random tie-breaking rule, after $t$ steps, for the same set of games, the estimates might be only within $O(1/t^{1/n})$ of NE, where $n$ is the number of possible strategies of one of the players.  In any case, the games that we consider do not necessarily have diagonal payoff matrices.  Moreover, there are algorithms that converge to NE much faster than FP \cite{DDK15}. Despite this, we focus on FP; our goal is not to find optimal algorithms, but to explain and justify human behavior.  FP has the advantage of being quite natural; it seems plausible that people would do something in the spirit of FP (which is partly why  it has attracted so much attention in the literature). 

To the extent that we model people as PFAs, they cannot implement full-blown FP, since it requires an unbounded number of states to keep track of the full history of play. Nevertheless, there are straightforward ways for a finite automaton to approximate FP; we start by considering that. Specifically, we approximate the probability that the opponent will go to a particular site $i$ by the fraction of times he has gone to site $i$ in the past $M$ steps. Instead of keeping a sliding window of size $M$, which requires keeping track of the order in which sites were visited, and thus requires quite a few states, each player keeps $M$ pieces of information in total and decides probabilistically what information to replace. Here $M$ can be viewed as a proxy for a player's memory capacity. At each step, both the poacher and the ranger play a best response to their beliefs of what their opponent is doing. More precisely, at each step, the player calculates the expected utility of going to each site (based on their estimate of the strategy of their opponent encoded in their memory) and plays the action with the highest expected utility, randomly choosing one if there are ties. 

We can modify Robinson's proof \cite{Robinson51} to show that if the PFA has sufficiently many states, players eventually converge to playing the NE strategies of the stage game and achieve the NE utilities. Since we are not particularly interested in the case where $M$ is large, rather than proving this result formally, we do simulations that illustrate this behavior. More interesting from our perspective is what happens if we limit the number of states of the PFA. Now the situation changes significantly; as we show, the PFA acts more human-like; the strategy it uses it somewhere between the NE strategy and \emph{probability matching} (i.e., visiting sites in proportion to the probability of a rhino being there).

Things become even more interesting if we do not treat all outcomes
equally. From a human perspective, some events are more significant
than others. Observing a potentially poisonous snake is far more
significant than observing a beetle.  We would expect humans to treat
significant events differently from less significant events.  In the
context of the ranger-poacher game, it seems reasonable to think that
a poacher views getting caught by the ranger as particularly
significant, because it gives him negative utility. (Although we do
not model poachers' snares being confiscated in our abstraction of the
game, Xu et al. \cite{XP20} observed that real-world poachers
reacted quite strongly to this event which, from their perspective, is
quite significant.) Similarly, the ranger views the poacher catching a
rhino as a significant event, because that is what gives her negative
utility. We can capture this significance by assuming that the poacher
and the ranger assign more weight to site $i$ in their memory if it is
related to an event with negative utility. Our simulations show that
taking significance into account, even in this naive way, can lead to
higher utility. We also show that the greater the weight of
significant events, the greater the improvement in the utility,
although the effect of the weightings has diminishing returns. We can
explain why this should be the case here: If the ranger or the poacher
use FP with a large memory, then the site that they consider best will
not change much, especially if they got rewarded by going there. If
the poacher overweights a site where he has been caught, he is less
likely to return there (and thus less likely to be caught by the
ranger, who will return there if she is using FP). We would expect
overweighting to have less of an impact the smaller the memory of the
poacher and ranger (since smaller memory makes returning to the same
site less likely), and this is indeed the case. Interestingly, Lieder,
Griffiths, and Hsu \cite{LGH18} have argued that
over-representation of extreme events leads to better decision-making
performance; our results suggest that this phenomenon is not just
limited to extreme events.  

To understand the effects of bounded memory and taking significance into account, for various choices of ranger strategy, we compared the performance (in terms of utility) of various parameter settings of our PFA to each other and to other poacher strategies. Our results showed that probability matching and overweighting significance can often lead to higher utility. This supports one of our key hypotheses: It can be quite rational to be (somewhat) ``irrational'', at least in the ranger-poacher game!

To see how humans actually play the ranger-poacher game, we ran experiments on Amazon Mechanical Turk (MTurk), using a number of different rhino distributions, with people playing the role of poacher for 100 rounds. The ranger in the experiment uses a PFA with $M = 100$ and $s=0$ (i.e., it does not take significance into account; we take $s=1$ if it does take significance into account).  When looking at the overall distribution (i.e., the fraction of times the human player visits each site), we found that we could roughly cluster players into three groups: \emph{level-0} poachers are non-strategic and either visit all sites with equal probability or stick to one site (they simply try to finish the game as quickly as possible); \emph{level-1} poachers best respond to level-0 rangers, which in our setting means that they probability match; and \emph{level-2} poachers best respond to level-1 rangers, which in our setting  means that they best respond under the assumption that rangers are probability matching. (These names are intentionally chosen to match the \emph{level-$k$} hierarchy of Stahl \cite{dale1993levelk}. Stahl defined level-0 players to be ones who choose a strategy at random, just as our level-0 poachers do, while level-($k+1$) players best respond under the assumption that they are playing against level-$k$ players.) Our experimental data show that most human poachers tend to probability match. As we show, this is also the case for poachers that use a PFA with small memory size $M$. As $M$ increases, our PFA will play a combination of probability matching and NE strategy, which is closer to the level-2 strategy. These and other observations suggest that modeling people as PFAs does capture important aspects of human behavior.

The rest of the paper is organized as follows: in Section~\ref{sec:game}, we formally define the ranger-poacher game as a normal-form game and prove that it has a unique NE. In Section~\ref{sec:pfa}, we define our PFA for the ranger-poacher game and show by simulations how it converges to NE strategies with sufficiently large memory and how it compares to other poacher strategies. We present the results from MTurk experiments and compare them with the simulation results in Section~\ref{sec:MTurk}. We conclude in Section~\ref{sec:discussion} with a discussion of related work and plans for future work.

\section{The Ranger-Poacher Game}\label{sec:game}
As we said, our ranger-poacher game is based on the wildlife poaching game of Kar et al. \cite{Kar15}.  There are two players, a ranger and a poacher, and a fixed number $n$ of sites that rhinos might go to.  We assume that the rhino distribution $d = (d_1, \ldots, d_n)$ is commonly known, where $d_i \in [0,1]$ is the probability that there is a rhino at site $i$ (we do not assume that $\sum_{i=1}^n d_i = 1$; there could be more than one rhino!). We denote by $\Gamma^K(d)$ the ranger-poacher game with rhino distribution $d$ and $K$ stages, whose stage game is denoted $\Gamma(d)$.  (Note that the distribution also implicitly encodes the number of sites.)  

Formally, the stage game $\Gamma(d)$ can be viewed as a normal-form game with players $P$ and $R$, whose expected payoff matrix is given in the following table, where the poacher is the row player and the ranger is the column player. 

\begin{center}
\begin{tabular}{ |c|c|c|c|c| } 
\hline
& 1 & 2 & \ldots & $n$ \\
\hline
1 & ($-1$, $1$) & ($d_1$, $-d_1$) & \ldots & ($d_1$, $-d_1$) \\ 
\hline
2 & ($d_2$, $-d_2$) & ($-1$, $1$) & \ldots & ($d_2$, $-d_2$)\\ 
\hline
\ldots & \ldots & \ldots & \ldots & \ldots\\
\hline
$n$ & ($d_n$, $-d_n$) & ($d_n$, $-d_n$) & \ldots & ($-1$, $1$)\\ 
\hline
\end{tabular}
\end{center}
We take the poacher's utility in $\Gamma^K(d)$ to be his average utility in each of the $K$ stage games, and similarly for the ranger. 

Although, in general, zero-sum games can have multiple equilibria, the ranger-poacher game has a unique NE.

\begin{proposition} For all $d$, $\Gamma(d)$ has a unique NE.
\end{proposition}

\begin{proof}  
We now show that the stage game $\Gamma(d)$ has a unique NE. Suppose that the poacher and the ranger use the mixed strategies $\sigma_P = (p_1, \ldots, p_n)$ and $\sigma_R = (r_1, \ldots, r_n)$, respectively, in the stage game, where $p_i$ is the probability that the poacher goes to site $i$ and $r_i$ is the probability that the ranger goes to site $i$, for all sites $i \in [1, \ldots, n]$. Let $visit^P_i$ and $visit^R_i$ denote the pure strategies where the poacher and the ranger, respectively to visit site $i$. Taking $u_P$ and $u_R$ to denote the poacher's  and ranger's expected utility, note that  
\begin{equation}\label{eq1}
\begin{array}{ll}
u_P(visit_i^P,\sigma_R) = (1-r_i)d_i - r_i \mbox{ and }\\
u_R(\sigma_P,visit_i^R) = p_i - \sum_{j \ne i} d_j p_j.\end{array}
\end{equation}
For $(\sigma_P,\sigma_R)$ to be a NE, the poacher has to be indifferent among all the pure strategies in the support of $\sigma_P$ if the ranger plays $\sigma_R$, and similarly for the  ranger. Thus, we require that
\begin{equation}\label{eq2}
\begin{array}{ll}
u_P(visit_i^P,\sigma_R) = u_P(visit_j^P,\sigma_R)
\mbox{ if  $p_i > 0$ and $p_j > 0$,}\\
u_P(visit_i^P,\sigma_R) \ge
u_P(visit_j^P,\sigma_R) \mbox{ if $p_i > 0$ and $p_j = 0$,}
\end{array}
\end{equation}
and similarly for the ranger.  Using (\ref{eq1}) and (\ref{eq2}), it is straightforward to solve for a NE of $\Gamma(d)$.  As we now show, the NE is unique.

It is well known that a two-player zero-sum game $\Gamma$ has a unique \emph{value} $v \ge 0$.  If each player $i$ plays a \emph{maximin} strategy (intuitively, a strategy that guarantees the best possible result for player $i$ under the assumption that the other player is trying to hurt him as much as possible), then one player will get an expected utility of $v$ and the other player will get  an expected utility of $-v$. Moreover, $(\sigma_1, \sigma_2)$ is a NE of $\Gamma$ if and only if $\sigma_i$ is a maximin strategy for player $i$, for $i =1,2$. It follows that, while there may be more than one NE in $\Gamma$, one of the players gets $v$ in all the NE, while the other player gets $-v$. (See, e.g., \cite{Ferguson20} for proofs of these results.)   

We show that there is in fact a unique NE in $\Gamma(d)$. Suppose, by way of contradiction, that there are two equilibria  $(\sigma_P,\sigma_R)$ and $(\sigma_P', \sigma_R')$, where $\sigma_P = (p_1, \ldots, p_n)^T$, $\sigma_P' = (p_1', \ldots, p_n')^T$, $\sigma_R = (r_1, \ldots, r_n)^T$, and $\sigma_R' = (r_1', \ldots, r_n')^T$. Further suppose that the value of the game is $v$.  Thus, we can assume without loss of generality that in all NE, the poacher gets expected utility $v$ and the ranger gets expected utility $-v$. Since all of $\sigma_P$, $\sigma_P'$, $\sigma_R$, and $\sigma_R'$ must be maximin strategies, it follows that $(\sigma_P,\sigma_R')$ and $(\sigma_P',\sigma_R')$ must also be NE. While $\sigma_P$ and $\sigma_P'$ may have different supports, the mixture $\frac{1}{2}\sigma_P + \frac{1}{2}\sigma_P'$ is also a maximin strategy, and its support is the union of that of $\sigma_P$ and $\sigma_P'$.  It follows for all pure strategies $visit^P_i$ and $visit^P_j$ in the support of either $\sigma_P$ or $\sigma_P'$, by (\ref{eq1}), we have that $-r_i + (1-r_i)d_i = u_P(visit^P_i,\sigma_R) = u_P(visit^P_i,\sigma_R') = -r_i' + (1-r_i')d_i$. Thus, $d_i -(1+d_i)r_i = d_i - (1+d_i)r_i'$, so $r_i = r_i'$. Thus, $\sigma_R = \sigma_R'$.   

We have to work a little harder to show that  $\sigma_P = \sigma_P'$. Note that $u_R(\sigma_P,visit^R_i) = u_R(\sigma_P,visit^R_j)$, so $$p_i - \sum_{k \ne i} d_k p_k = p_j - \sum_{k \ne j} d_k.$$ Similarly, $u_R(\sigma_P',visit^R_i) = u_R(\sigma_P',visit^R_j)$, so $$p_i' - \sum_{k \ne i} d_k p_k' = p_j' - \sum_{k \ne j} d_k p_k'.$$ Simple algebra shows that $$(p_i - p_j) + d_ip_i - d_jp_j = 0 = (p_i' - p_j') + d_ip_i' - d_j p_j'.$$ Therefore, 
\begin{equation}\label{eq0}
(1+d_i)(p_i -p_i')  = (1+d_j)(p_j - p_j').
\end{equation}
If $\sigma_P \ne \sigma_P'$, then there must be some site $i$ in the support of $\sigma_P$ such that $p_i > p_i'$.  Thus, $(1+d_i)(p_i - p_i') > 0$.  But then it follows from (\ref{eq0}) that $p_j > p_j'$ for all sites $j$ in the support of either $\sigma_P$ or $\sigma_P'$. Since $\sum_i p_i = \sum_i p_i' = 1$, this is clearly a contradiction. Thus, $p_i = p_i'$ for all sites $i$, so $\sigma_P = \sigma_P'$.  
\end{proof}

It now easily
follows by backward induction on 
$K$ that the only NE in $\Gamma^K(d)$ is for the ranger and poacher to
play the NE in each stage game. Thus, the poacher's expected utility
in $\Gamma^K(d)$ if the NE is played is exactly the same as his
expected utility in the stage game $\Gamma(d)$ if the NE of that game
is played, and similarly for the ranger.  

It is not hard to construct distributions such that the poacher's expected utility in the NE of $\Gamma(d)$ is negative. Typically, this happens if the probability of finding a rhino is low. To take a simple example, in the extreme case where $d_i = 0$ for all $i$,  since there is a positive probability that the poacher is caught, no matter what strategy he uses, his expected utility will be negative. Of course, in the real world, if this were the situation, poachers simply wouldn't poach at all.  For simplicity, we have not given the poacher an action of ``no poaching'' (with a payoff of 0 to both the ranger and poacher if it is played). In our experiments, we consider only distributions $d$ such that the poacher has positive expected utility in the NE of $\Gamma(d)$.    

\section{PFAs Play the Ranger-Poacher Game}\label{sec:pfa}
In this section, we present a family of PFAs for the ranger-poacher
game, and discuss their performance; in the next section, we compare
their performance to that of people playing the ranger-poacher game. 
A PFA is just like the deterministic finite automaton, except its action and transition functions are probabilistic. Since we want our automata to produce an output (a site to visit), rather than accepting a language, we are technically looking at what have been called \emph{probabilistic finite automata with output}. Formally, we can define a PFA with output as a tuple \((Q, q_0, \Sigma, O, \gamma, \delta)\), where   
\begin{itemize}
     \item $Q$ is a finite set of \emph{states};
     \item $q_0 \in Q$ is the initial state;
     \item $\Sigma$ is the input alphabet (in the case of the
      ranger-poacher game, this will consist of the observations
``visited site $i$, the opponent visited site $j$,  
and received utility $u$'' for $i,j \in\{ 1, 
      \ldots, n\}$ and $u \in\{-1, 0,
      1\}$);  
     \item $O$ is the output alphabet (in our case this will be
``$i$'', which is interpreted as visiting site $i$, for $i \in\{
     1, 
      \ldots, n\}$);  
     \item $\gamma: Q \rightarrow \Delta(O)$ is a probabilistic action function (as usual, $\Delta(X)$ denotes the set of probability distributions on $X$);
     \item $\delta : Q \times \Sigma \rightarrow \Delta(Q)$ is a probabilistic transition function.
 \end{itemize} 
Intuitively, the automaton starts in state $q_0$ and visits a site according to distribution $\gamma(q_0)$.  It then observes the outcome $o$ of visiting the site (an element of $O$) and transitions to a state $q'$ (according to $\delta(q_0,o)$).  It then visits a site according to the distribution $\gamma(q')$, and so on. 

\subsection{Fictitious play: a review} \label{sec:pfa review}
As  we said in the introduction, with FP, the players best respond at each step to the strategy their opponent has played so far.  To make this precise, given a two-player game, where player 1 has $m$ possible strategies and player 2 has $n$ possible strategies, which we take to be $\{1,\dots, m\}$ and $\{1, \ldots, n\}$, we define a sequence of mixed strategies $x^1, x^2, \ldots$ and $y^1, y^2, \ldots$ for players 1 and 2, respectively, and sequences of pure strategies $s_1, s_2, s_3, \ldots$ and $s'_1, s'_2, s'_3, \ldots$ for players 1 and 2, respectively, as follows.  (Technically, these are all random variables, not (mixed/pure) strategies.)  Let $s_1$ be a random element of $\{1, \ldots, m\}$ and let $s_1'$ be a random element of $\{1, \ldots, n\}$ (both chosen uniformly at random).  Then let $x^t(i)$ be the fraction of times that strategy $i$ appears in $s_1, \ldots, s_t$, let $y^t(j)$ be the fraction of times that strategy $j$ appears in $s'_1, \ldots, s'_t$, let $s_{t+1}$ be a best response to $y^t$ (where if there are ties, $s_{t+1}$ is chosen at random), and let $s'_{t+1}$ be a best response to $y^t$ (where again, ties are broken at random).  This completes the description of FP. 

\subsection{The poacher's PFA} \label{sec:poacher's pfa}
Formally, we consider a family of possible PFAs for the poacher that have the form $A_{n, d, M, s} = (Q_{n,M}, q_{0},\Sigma_n, O_n, \gamma_{n, d, s},$ $\delta_{n,M,s})$, where there are 4 parameters: $n$ is the total number of sites, $d$ is the rhino distribution over sites, $M$ is a bound on the poacher's memory, and $s$ is either 0 or 1, depending on whether we want to take significance into account. We assume that $n$ and $d$ are given as part of the description of the game; we choose $M$ and $s$ (and will examine the effect of different choices). In more detail, the components of the tuple are as   follows:  
\begin{itemize}
\item A state $q \in Q_{n, M}$ has the form $([q_1,  \ldots, q_n],i)$,
where $q_1 + \cdots + q_n \le M$, as discussed earlier, and
$i \in \{1,\ldots,n\}$ is a best response for the poacher with respect
to the ranger strategy $(q_1/K,  \ldots, q_n/K)$, where $K = q_1
+ \cdots + q_n$.  (If $K = 0$, then $i=1$.) Thus, there are $n
{M+n \choose n}$ possible states. 

\item The initial state $q_0 = ([0, \dots, 0],1)$.

\item $\Sigma_n$ consist of observations of the form $(j,u) \in \{1,\ldots,n\} \times \{-1,0,1\}$; intuitively, $j$ is the site where where the ranger last went and $u$ is the poacher's utility in the last round. 

\item $O_n = \{1, \ldots, n\}$: the poacher can visit any site.
  
\item The action function $\gamma_{n, d, s}$ goes to site $j$ in a state of the form $(\cdot,j)$.

\item The transition function $\delta_{n,M,0}$ proceeds as follows. In state $([q_1, \ldots, q_n],i)$, given the  observation $(j,u)$, the new state has the form $([q_1', \ldots, q_n'],i')$, where $i'$ is a best response (chosen at random among all the best responses) to the strategy $[q_1'/K', \ldots, q_n'/K']$, where $K' = q_1' + \cdots+ q_n'$, and $[q_1', \ldots, q_n']$ is computed using the following procedure: If $K = q_1 + \cdots +q_k < M$, then $q_j' = q_j + 1$ and $q'_{j'} = q_{j'}$ if $j' \ne j$, and if $K \ge M$, then $j'$ is chosen at random according to the distribution $[q_1/K, \ldots, q_n/K]$. If $j' \ne j$, $q'_{j'} = q_{j'}-1$, $q_j' = q_j +1$, and $q_{j''}'= q_{j''}$ for $j'' \notin \{j,j'\}$, whereas if $j'=j$, then $q_{j''}' = q_{j''}$ for all $j''$. The transition function $\delta_{n,M,1}$ proceeds the same way, except in the case that $u=-1$ (which means that the poacher was caught).  In this case, rather than choosing one site at random to decrement, two sites are chosen, and $q_j$ is increased by 2.

\end{itemize}

The ranger's PFA is similar in spirit. The idea is straightforward: If the poacher knew the fraction $r_i$ of times that the ranger has visited each site $i$ by time $t$, he could construct the ranger's strategy $y^t$ by taking $y^t(i) = r_i$. He could then calculate a best response to $y^t$ by calculating the utility of visiting each site, and then choosing the site that gives the highest utility, randomizing in the case of ties. Since calculating $r_i$ exactly requires unbounded memory, we instead approximate it in a way that uses relatively few states.

Suppose there are a total of $n$ sites and the poacher can remember up to $M$ events.  Suppose that the game has been played for $M'$ rounds. Roughly speaking, we want the poacher to remember what happened in the past $M$ rounds (or the last $M'$ rounds if $M' \le M$). We take the poacher's memory to have the form $[q_1, \ldots, q_n]$, where $q_i$ is an approximation to the number of times that the ranger has visited site $i$ in the past $M$ rounds (or the exact number of times that the ranger has visited site $i$ in the past $M'$ rounds, if $M' \le M$). Note that there are ${M + n \choose n}$ memories of this form, since we can associate a state with a string of $n$ 0s and $M$ 1s: $q_1$ is the number of 1s to the left of the first 0; for $2 \le i \le n$, $q_i$ is the number of 1s between the $(i-1)$st and $i$th 0s; and the number of 1s after the last 0 is $\max(M-M',0)$.\footnote{We could have also taken a state to be a sliding window of the last $M$ observations made by the poacher. But then there would be $n^M$ possible states.  Since in our cases $n$ is small (typically 3--5), while $M$ is relatively large (10--199), $n^M$ is significantly larger than ${M+n \choose n}$.}     

We update the memory using ideas similar to those used by Liu and
Halpern \cite{Liu2020MAB}. As long as the memory is not full
(i.e., as long as $q_1 + \cdots + q_n \le M$),  if the ranger goes to
site $j$, then $q_j$ is increased by 1. Once the memory is full, the
poacher first decreases some $q_{j'}$ chosen at random according to
(the poacher's estimate of) the ranger's current strategy by 1 (e.g.,
if the current strategy is  $(0,0.4,0.6)$, then $q_2$ is decreased by
1 with probability 0.4 and $q_3$ is decreased by 1 with probability
0.6), then $q_j$ is increased by 1. As we show by simulation, this
approach to approximating the ranger's strategy is reasonably
accurate. In Figure~\ref{fig:1}, we fix the ranger's strategy to be
$(0.2, 0.3, 0.5)$ and compare that with the poacher's estimate for the
first 1000 steps, for $M=100$ and $M=10$.
(In the figure, the ``actual lines'' represent the frequencies of
visits by the ranger; the ``estimate lines'' are the poacher’s estimate
of those frequencies, given his memory limitations.) 
In both cases, the estimate
fluctuates around the true strategy, with the amount of fluctuation
increasing as the $M$ decreases, as we would expect. The results are
similar for other rhino distributions. 

\begin{figure}[htb]
    \centering
    \includegraphics[width = 8.5cm]{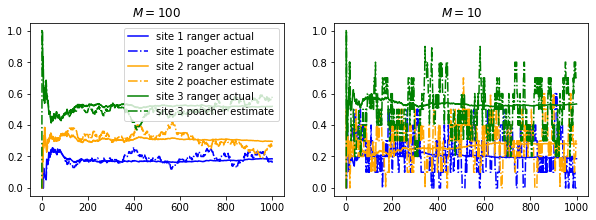}
    \caption{Comparing the poacher's estimation to the ranger's actual strategy for various memory sizes.}  
    \label{fig:1}
\end{figure}

To take significance into account, if the ranger and the poacher both
go to site $i$ (so the poacher was caught), then $q_j$ and $q_{j'}$
are decreased by 1 for two sites $j$ and $j'$ chosen at random
according to the current strategy, instead of just one site (we may
have $j=j'$, in which case $q_j$ is decreased by 2), and $q_i$ is
increased by 2 (rather than 1).  By doing this, the poacher is
effectively overweighting sites where he was caught by the ranger. 

To further investigate the effect of taking significance into account,
we fixed both the poacher and ranger to use a PFA with $M = 1000$ for
games with 1000 rounds and various rhino distributions. The parameter
$s$ indicates the weighting for significance; the poacher does not
take significance into account when $s=0$. When $s=1$, the poacher
overweights significant events by increasing the count by 2 (instead
of 1 if $s=0$). Similarly, the poacher will increase the count by
$s+1$ for other values of $s$. In this case, we can control the
weighting for significance through parameter $s$. Table~\ref{sig}
shows the poacher's utilities for various significance weightings and
rhino distributions. All results are averaged over 100
repetitions. Our simulations show that the greater the weight, the
greater the improvement in the poacher’s utility, although the effect
has diminishing returns. 

\begin{table}[htb]
\centering
\begin{tabular}{ l l l l l l}
\hline
rhino distribution           & $s=0$ & $s=1$ & $s=2$ & $s=3$ & $s=4$ \\ \hline
$(0.2, 0.4, 0.6, 0.8)$       & 0.098 & 0.165 & 0.188 & 0.203 & 0.205 \\ \hline
$(0.3, 0.8, 0.7, 0.5)$       & 0.165 & 0.233 & 0.259 & 0.271 & 0.275 \\ \hline
$(0.9, 0.9, 0.9)$            & 0.267 & 0.375 & 0.415 & 0.425 & 0.435 \\ \hline
$(0.9, 0.6, 0.4, 0.9)$       & 0.509 & 0.605 & 0.626 & 0.641 & 0.652 \\ \hline
\hline
\end{tabular}
\caption{Poacher's utility for various significance weightings under different rhino distributions.}
\label{sig}
\end{table}

Intuitively, taking significance into account
improves the poacher's utility because by overweighting sites
where the poacher was caught by the ranger, the poacher is more likely
to avoid going to
those sites in future rounds, and hence be less likely to get
caught.
Roughly speaking, by overweighting, the poacher is doing the right
thing faster.  This intuitions is reinforced by our simulations, which
show that, not only does the poacher get higher
utility by taking significance into account, but the poacher converges
faster to the NE 
strategy.
We are currently working on formalizing these observations (see
Section~\ref{sec:discussion}).

\subsection{The effect of different memory sizes} \label{sec:memory}
As we said in the introduction, Robinson \cite{Robinson51} showed
that in a two-player zero-sum game where both players have only a
finite number of strategies, FP converges to NE.  That is, the
frequency that the players go to site $i$ over the first $N$ rounds
converges to the probability of going to site $i$ according to the
NE. The utilities also converge to these in NE. Not surprisingly, we
also get approximate convergence using our PFA with sufficiently many
states. We did simulations with various rhino distributions, for games
with 1000 rounds. Figure~\ref{fig:2} describes the results for the
rhino distribution $(0.2, 0.4, 0.6, 0.8)$ (the results were the same
for all other distributions we considered). As shown in
Figures~\ref{fig:2}(a) and \ref{fig:2}(b), we get quite a good
approximation with $M = 1000$ since FP converges to NE
strategies. However, as the memory size decreases, the poacher gets
closer and closer to probability matching.  (As we discuss below, the
ranger is close to probability matching all along.) As shown in
Figure~\ref{fig:2}(d), if $M=100$, after about 600 steps, the
frequency with which the ranger goes  to site 1 is almost identical to
the frequency with which the ranger goes to site 1 in NE.  As shown in
Figure~\ref{fig:2}(f), even with $M=10$, the ranger is quite close to
NE.  On the other hand, as shown in Figure~\ref{fig:2}(e), with
$M=10$, the poacher is much closer to probability matching. Actually,
as we now explain, the behavior of both the poacher and the ranger can
be thought of as a convex combination of probability matching and
playing the NE. 

Note that when the ranger plays her part of the NE, all the poacher's
responses are equally good. This might suggest that if the poacher is
best responding, he should just go to all sites with equal
likelihood. This clearly does not happen. A closer look explains
why. As shown in Figure~\ref{fig:1}, the poacher's estimation of the
ranger's strategy will fluctuate around the actual strategy. The
smaller $M$ is, the greater the fluctuation. If the fluctuation at
round $t$ is such that only one site gets lower probability than it
does in NE, then the best response for the poacher at round $t$ is to
go to that site. But if at round $t$ there are two sites, say $j$ and
$j'$, such that the ranger goes to sites $j$ and $j'$ with frequency
less than in NE (according to the poacher's estimate), then the
poacher's best response depends on how much below the NE strategy the
frequency of the ranger going to each of site $j$ and $j'$ is, and the
probability of finding rhinos at sites $j$ and $j'$.  For example, an
equal decrease in the probability of going to $j$ and $j'$ makes the
site with higher rhino probability the best option.  In particular,
this means that symmetric fluctuations result in the poacher
preferring sites with higher rhino probabilities, roughly in
proportion to their respective probability. Thus, decreasing memory
results in a greater likelihood of probability matching for the
poacher, even though the poacher continues to behave rationally (i.e.,
by best responding to his estimate)!
Put another way, rational but resource-bounded agents will act like
probability matchers, and the more resource-bounded they are, the more
marked this behavior will be.

This effect is less marked with the ranger. To see why, note that with the rhino distribution $(0.2, 0.4, 0.6, 0.8)$, the ranger's strategy in NE, $\sigma_R^* = (0.08, 0.22, 0.31, 0.39)$, is quite close to the probability-matching strategy, $(0.1, 0.2, 0.3, 0.4)$. This is a general phenomenon in the ranger-poacher game, and not just an artifact of this rhino distribution.  Thus, even though the ranger's strategy is also affected by probability matching as memory size decreases, the effect is not as noticeable. (We remark that Erev and Barron \cite{EB05} also pointed out the connection between probability matching and small samples in a simpler setting.)

\begin{figure}[htb]
    \centering
\includegraphics[width = 8.5cm]{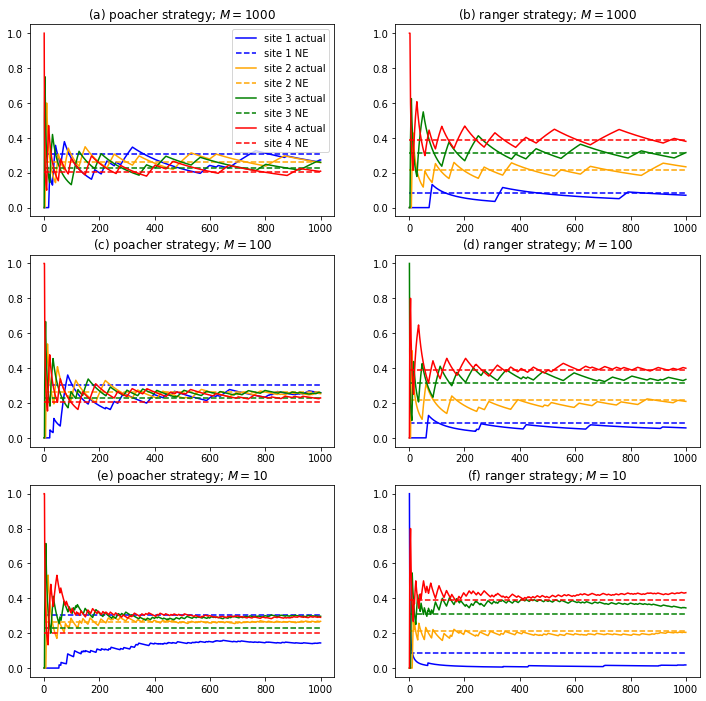}
\caption{PFA strategies with various memory sizes.}
    \label{fig:2}
\end{figure}

\subsection{Comparing strategies} \label{sec:strategies}
To understand the effects of memory and taking significance into
account, we compared the performance (in terms of utility) of various
parameter settings of our PFA to each other and to other poacher
strategies for various choices of ranger strategies. We considered
eight poacher strategies: (1) the NE strategy, which can be viewed as
a baseline; (2) FP with unbounded memory (although this leads to the
NE strategy if the ranger also uses it, it does not in general); (3)
\emph{multiplicative weight updating} (MWU) \cite{AHK12}, a strategy
that has been shown to lead to NE quickly; (4) \emph{utility matching}
(UM) (instead of best responding, a site is chosen with probability
proportional to its expected utility); (5) PFA1: a PFA with limited
memory and no overweighting of significant events ($M = 100, s=0$);
(6) PFA2: a PFA with limited memory that overweights significant
events ($M = 100, s=1$); (7) PFA3: a PFA with very limited memory and
no overweighting ($M = 10, s=0$); (8) PFA4: a PFA with very limited
memory that overweights significant events ($M = 10, s=1$). We want to
see how each of these eight poacher strategies plays against the
various ranger strategies. We consider four ranger strategies: (a) the
NE strategy; (b) probability matching (PM) based on the rhino
distribution; (c) FP with unbounded memory; (d) PFA with small memory
and no overweighting ($M = 10, s=0$). Notice that ranger strategies
(a) and (b) are nonadaptive; the ranger visits a site according to a
predetermined distribution at each step.  In contrast, strategies (c)
and (d) are adaptive; The ranger decides which site to visit next
based on what the poacher has done in previous rounds. For each
ranger-poacher pair, we simulated the game for 1000 rounds, using
various rhino distributions, and repeated each game 100
times. Figure~\ref{fig:3} shows the boxplots \footnote{A boxplot is a
standard way of displaying the dataset. The box in the middle shows
the range from the lower quartile (25$^{th}$ percentile) to the upper
quartile (75$^{th}$ percentile), with the red line marking the
median. The \emph{interquartile range} (IQR) is the difference between
values at the 75$^{th}$ and 25$^{th}$ percentile.  A whisker is drawn
from the top of the box to the largest observed point that is within
1.5 times the IQR of the value at the 75$^{th}$ quartile.  Similarly,
a whisker is drawn from the bottom of the box to the smallest observed
point that is within 1.5 times the IQR of the value at the 25$^{th}$
quartile. (See \citet[pp. 234--238]{Dekking05}.)} for the rhino
distribution $(0.2, 0.4, 0.6, 0.8)$. 

We get similar results for other rhino distributions. Figures~\ref{fig:5} and ~\ref{fig:6} show the boxplots of poacher's utilities for various poacher and ranger strategies for rhino distributions $(0.3, 0.8, 0.7, 0.5)$ and $(0.9, 0.9, 0.9)$, respectively. 

\begin{figure}[htb]
    \centering
    \includegraphics[width=8.5cm]{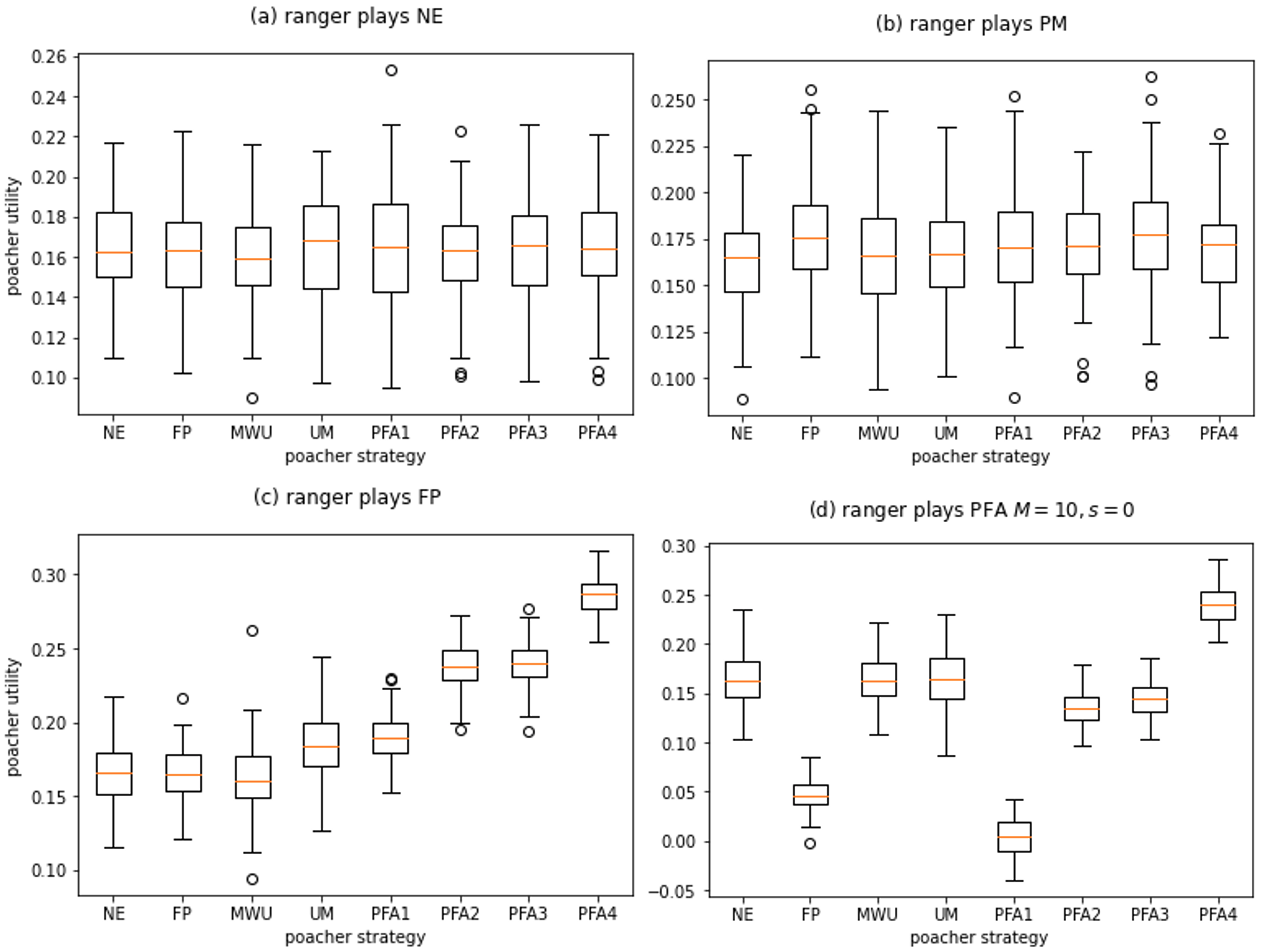}
    \caption{Poacher's utility for various poacher and ranger strategies and rhino distribution (0.3, 0.8, 0.7, 0.5).} 
 \label{fig:5}
\end{figure}

\begin{figure}[htb]
    \centering
    \includegraphics[width=8.5cm]{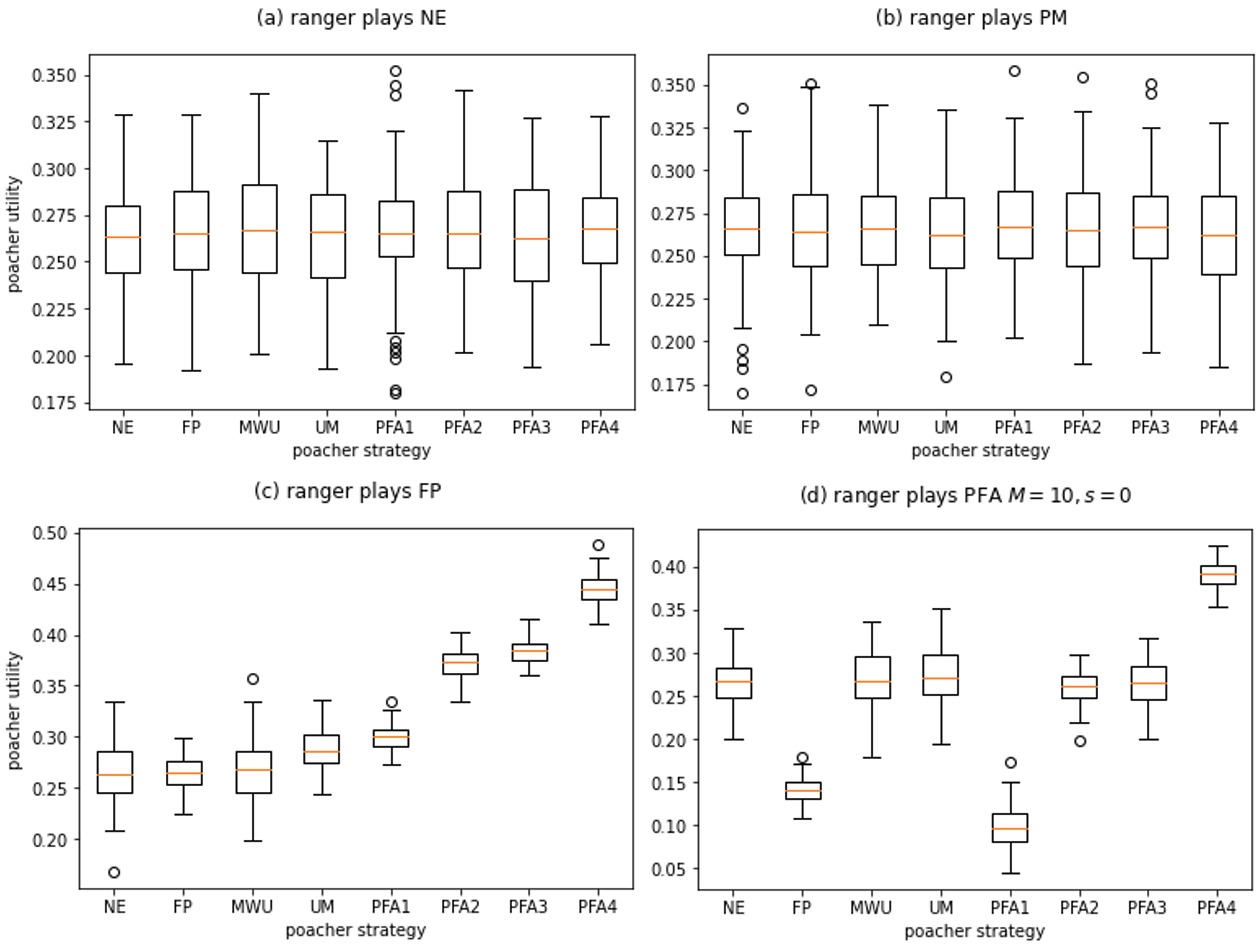}
    \caption{Poacher's utility for various poacher and ranger strategies and rhino distribution (0.9, 0.9, 0.9).} 
 \label{fig:6}
\end{figure}

The results show that if the ranger uses a nonadaptive strategy (NE or PM), then all the poacher's strategies do equally well, whereas if the ranger uses an adaptive strategy (FP or PFA), a PFA with small memory size that overweights significance (PFA4) does extremely well.

\begin{figure}[htb]
    \centering
    \includegraphics[width=8.5cm]{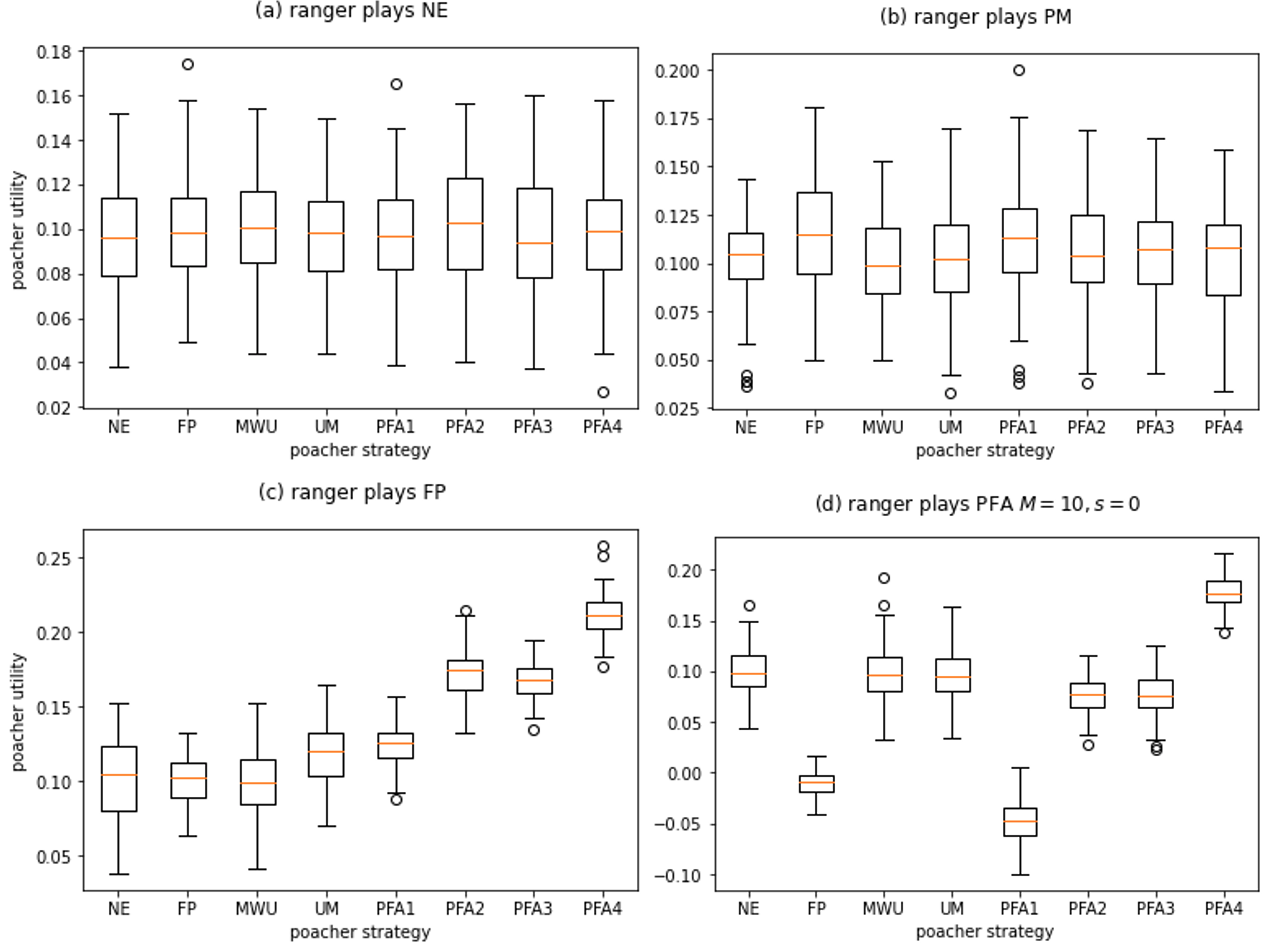}
    \caption{Poacher's utility for various poacher and ranger strategies and rhino distribution (0.2, 0.4, 0.6, 0.8).} 
 \label{fig:3}
\end{figure}

As shown in Figures~\ref{fig:3}(a) and ~\ref{fig:3}(b), when the ranger plays a non-adaptive strategy, all the poacher strategies do roughly equally well.  That means, for example, the poacher does not suffer due to the reduced memory size of 10 when the ranger plays the NE strategy; he still gets essentially the same payoff as he would when playing the NE strategy, or using a strategy like MWU or FP that requires unbounded memory. However, as shown in Figure~\ref{fig:3}(c) and ~\ref{fig:3}(d), using a PFA with limited memory can significantly \emph{improve} performance, especially when significance is taken into account, if the ranger's strategy is adaptive. It is not hard to explain why this should be so: If the ranger is playing FP with unlimited memory, she tends to be ``sticky''; if she views site $i$ as the best site at time $t$, she is likely to still view $i$ as the best site at time $t+1$. If the poacher was caught at site $i$ at time $t$, he is more likely to avoid $i$ if he overweights being caught (i.e., if $s=1$), and if he does some probability matching rather than just going to what he views as the best site. Thus, far from being irrational, overweighting significant events and probability matching are completely rational if the ranger is using FP. If the ranger is also using PFA with a relatively small memory size, the ranger will also switch more often, so using FP or PFA with a larger memory (PFA1) will make the poacher stickier and lead to lower utility for the poacher. However, overweighting significance and probability matching still help improve the poacher's utility significantly. 

\section{Experiments}\label{sec:MTurk}
We  wanted to understand the extent to which our PFAs capture human behavior in the ranger-poacher game. We conducted experiments on MTurk. In these experiments, human subjects play the role of the poacher; they must decide which site to visit in each round. Each game lasts for 100 rounds. Subjects get \$1 for completing the task plus a bonus of \$0.10 for each point they obtain. Since the game lasts for 100 rounds, the bonus is usually significantly more than the fixed payment. Thus, they are (somewhat) incentivized to maximize their payoff by playing strategically. 
We submitted an IRB consent form and qualified for exemption from IRB
review.

Subjects are given the rhino distribution and are told that the ranger
knows it as well. They are also told that, in each round, they and the
ranger will simultaneously choose a site to visit.  After these
choices are made, they are revealed, and the subjects discover where
the rhinos were, so they can see whether they caught a rhino or were
caught.  They get $1$ point if they catch a rhino without being
caught, $-1$ point if they are caught, and $0$ points otherwise.
See the supplementary material for the exact instructions and screenshots of the game interface.

The experiments involved 94 participants, each playing two
100-round games with two rhino distributions. We considered four rhino
distributions: (a) $(0.9, 0.6, 0.2)$, (b) $(0.9, 0.6, 0.4, 0.9, 0.1)$,
(c) $(0.8, 0.3, 0.8, 0.3)$, and (d) $(0.3, 0.8, 0.7, 0.5)$. We used
the PFA to play the role of the ranger, with $M = 
100$ and $s=0$ (i.e., it does not take significance into account).
$s=0$.

Since there might be some learning during the game, we treat the first 75 of the 100 rounds as the learning period, and average only the decisions in the last 25 rounds to get the poacher's overall strategy. Applying $k$-means clustering, we clustered players on  MTurk into three groups: (1) level-0 poachers, who visit all sites with equal probability or simply stick to one site; (2) level-1 poachers, who visit each site with probability roughly proportional to the rhino distribution; and (3) level-2 poachers, who seem to visit sites with probability proportional to their utility under the assumption that the rangers are playing a level-1 strategy.

We suspect that level-0 players are often ones who simply want to get the game done as quickly as possible, so that they can collect the fixed payment. Therefore, we focus on level-1 and level-2 players. We can best approximate level-1 poacher behavior using a PFA with $M=2$ and $s=1$; as explained earlier, a PFA that has a small memory will do more probability matching. We can best approximate level-2 poacher with a PFA with $M=10$ and $s=0$. The comparisons in overall strategy are illustrated in Figure~\ref{fig:4} for rhino distribution $(0.8, 0.3, 0.8, 0.3)$.

Each PFA was simulated 100 times. In the figure, we again use boxplots
to describe the distributions. Since the means here are occasionally
significantly different from the medians in these examples,   we
include the means as well as the medians (using a green triangle for
the mean). As the figure shows, our PFAs do a relatively good job of
capturing these two types of behaviors, despite using only two
features: the memory size and whether we take significance into
account.   

\begin{figure}[htb]
    \centering
    \includegraphics[width=8.5cm]{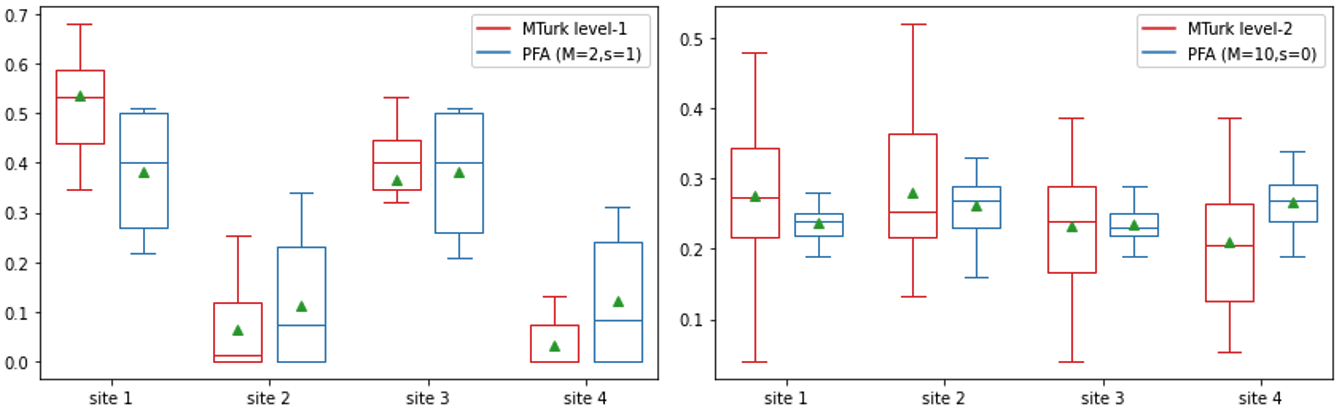}
    \caption{Comparing level-1 (left) and level-2 (right) behaviors (in frequency of visiting each site) between humans and PFAs for rhino distribution (0.8, 0.3, 0.8, 0.3).} 
 \label{fig:4}
\end{figure}

In addition to the overall strategy, we also considered another metric of decision-making dynamics: stickiness. Our experiments show that people were stickier when they obtained 0 or 1 point, compared to when they get caught; we observe qualitatively similar behavior with PFAs. 
We define stickiness as the probability of visiting same site at round $t+1$  as at round $t$. Stickiness is inherent in FP; if going to site $i$ was a best response for the poacher, and the poacher caught a rhino there without being caught, then going to site $i$ will continue to be a best response.  On the other hand, if the poacher is caught at site $i$, this will increase the poacher's estimate of the ranger being at site $i$ (particularly if we take significance into account), and thus make it less likely that the poacher will return to site $i$. We consider how sticky humans and the PFA were as a function of the reward received at round $t$. In Table~\ref{t2}, we illustrate the degree of stickiness of the poacher. Again, we played the game for 100 rounds and took the ranger's strategy to be a PFA with $M=100$ and $s=0$. As shown in Table~\ref{t2}, people  were  stickier  when  they obtained 0 or 1 point than when they got -1 points; we observe qualitatively similar behavior with our PFAs. 

\begin{table}[htb]
\centering
\begin{tabular}{l l l l l}
\hline
                 & -1 points & 0 points &  1
                 point \\ \hline 
Humans           & 0.3546 & 0.3820 & 0.4936 \\ \hline
PFA$_{M=2,s=0}$ & 0.0040 & 0.7479 & 0.8756 \\ \hline
PFA$_{M=2,s=1}$ & 0.0 & 0.7416 & 0.8653 \\ \hline
PFA$_{M=10,s=0}$ & 0.4323 & 0.8488 & 0.9238 \\ \hline
PFA$_{M=10,s=1}$ & 0.1021 & 0.8250 & 0.9117 \\ \hline
\hline
\end{tabular}
\caption{The probability that the poacher goes to same site at round
$t+1$ conditional on the poacher receiving $-1$ points (resp., 0
points, 1 point)  at round $t$.} 
\label{t2}
\end{table}

\section{Discussion and Related Work}\label{sec:discussion}
There has been a great deal of recent interest in modeling human
behavior in the ranger-poacher game. Perhaps most relevant to us is a
sequence of papers by Tambe and his collaborators. Nguyen et
al. \cite{Nguyen2013SUQR} proposed the \emph{Subjective Utility
Quantal Response} (SUQR) model. This model already allows for poachers
that are boundedly rational, in the sense of not necessarily best
responding.  It assumes that a poacher's subjective utility is
characterized by various features (e.g., the ranger’s coverage
probability, the poacher’s own reward and penalty at each target), and
tries to learn  the weights of these features from available attack
data so as to predict the behavior of poachers. It further assumes
that there is a homogeneous population of poachers; Yang et
al. \cite{yang2014PAWS} refined the model by allowing different
poachers to be characterized by different weight vectors. Kar et
al. \cite{Kar15} further refined the model by considering
successes and failures of the poacher’s past actions. Feng et
al. \cite{FST15} considered poachers with a fixed memory
$\Gamma$, which is the number of rounds of past observations they
respond to, similar to our use of $M$ in this paper. Kar et
al. \cite{Kar17} presented a behavior model based on an ensemble of decision trees. The goal of all this work was essentially to learn and predict human behavior.  This is part of a more general thrust of trying to model human play in games (see, e.g., \cite{wright2010beyond,WLB17,WLB19}). By way of contrast, our goal is to see the extent to which we can \emph{explain} and \emph{justify} apparently irrational human behavior (like probability matching and overweighting of significant events) as the outcome of computational limitations. To that end, we model resource-bounded poachers and rangers as PFAs. We showed that quite rational behavior (i.e., best responding) can lead to behaviors that have been viewed as irrational, namely, probability-matching and overweighting, as we limit the memory size. However, our results show that this so-called irrational behavior actually leads to \emph{better} outcomes. 

In future work, we would like to consider the effect of another
trait: stubbornness.  We can model this by assuming that there is a
stubbornness parameter $\alpha \in [0,1]$: with probability $\alpha$,
the poacher ignores the current information (regarding where the
ranger went), and just plays according to (his estimate of) his own
current strategy. It seems reasonable to expect people to grow more
stubborn over time. After all, after they have played for a while,
they feel that they understand what is going on and have already
settled on a strategy. There is no point in putting in the cognitive
effort of looking at what the ranger is doing and updating.  Indeed,
we find that if we slowly increase the poacher's stubbornness over
time, we get better behavior: the strategy curve is smoother, and
fluctuates less around NE. Of course, if we increase it too quickly,
then the poacher will settle into a ``bad'' strategy and not converge
to NE.  Moreover, using stubbornness requires the poacher to keep
track of his own strategy, which squares the number of states
needed. We hope to explore the tradeoffs involved with considering
stubbornness in future work.  

The fact that computational limitations lead both to more human-like
behavior and (often) to better outcomes in the ranger-poacher game
reinforces similar results obtained in other contexts
\cite{W02,HPS12,Liu2020MAB}.  
This suggests a rather rich line of future research.  For one thing,
it would be of interest to see if behavior in other games, such as
coordination games, can be explained and justified by computational
limitations.  For example, people are quite good at
coordinating, but do not always manage.  We suspect that computational
limitations play a fundamental role here.  Understanding the effect of
computational limitations better in a number of games of interest
better would lead to the design of better mechanisms for these games.
Moving up a level, getting such a deeper understanding might
lead a more general theory formalizing when and the extent
to which behaviors such as probability matching are consequences of
computational limitations and when behaviors such as overweighting
lead to better outcomes.  For example, we gave a somewhat informal
argument earlier for  why we see probability matching in the
ranger-poacher game.  It should not be hard to formalize this, but it
would be more interesting to get a general result explaining when we
would expect to see probability matching as an outgrowth of
computational limitations.   We look forward to investigating and
reporting on these issues in the future.

\newpage

\bibliographystyle{chicago}
\bibliography{z,joe,kash}

\begin{thebibliography}{}

\bibitem[\protect\citeauthoryear{Abernethy, Lai, and Wibisono}{Abernethy
  et~al.}{2019}]{ALW19}
Abernethy, J., K.~A. Lai, and A.~Wibisono (2019).
\newblock Fast convergence of fictitious play for diagonal payoff matrices.
\newblock Available at http://arxiv.org/abs/1911.08418.

\bibitem[\protect\citeauthoryear{Ariely}{Ariely}{2010}]{Ariely10}
Ariely, D. (2010).
\newblock {\em Predictably Irrational: The Hidden Forces That Shape Our
  Decisions}.
\newblock New York, NY: Harper Perennial.

\bibitem[\protect\citeauthoryear{Arora, E.Hazan, and Kale}{Arora
  et~al.}{2012}]{AHK12}
Arora, S., E.Hazan, and S.~Kale (2012).
\newblock The multiplicative weights update method: {A} meta-algorithm and
  applications.
\newblock {\em Theory of Computing\/}~{\em 8}, 121--164.

\bibitem[\protect\citeauthoryear{Brown}{Brown}{1951}]{Brown51}
Brown, G.~W. (1951).
\newblock Iterative solutions of games by fictitious play.
\newblock In T.~C. Koopmans (Ed.), {\em Activity Analysis of Production and
  Allocation}, pp.\  374--376. New York: Wiley.

\bibitem[\protect\citeauthoryear{Daskalakis, Deckelbaum, and Kim}{Daskalakis
  et~al.}{2015}]{DDK15}
Daskalakis, C., A.~Deckelbaum, and A.~Kim (2015).
\newblock Near-optimal no-regret algorithms for zero-sum games.
\newblock {\em Games and Economic Behavior\/}~{\em 92}, 327--348.

\bibitem[\protect\citeauthoryear{Daskalakis and Pan}{Daskalakis and
  Pan}{2014}]{DP14}
Daskalakis, C. and Q.~Pan (2014).
\newblock A counter-example to {K}arlin's strong conjecture for fictitious
  play.
\newblock In {\em Proc.~55th IEEE Symposium on Foundations of Computer
  Science}, pp.\  71--78.

\bibitem[\protect\citeauthoryear{Dekking}{Dekking}{2005}]{Dekking05}
Dekking, F.~M. (2005).
\newblock {\em A Modern Introduction to Probability and Statistics}.
\newblock Springer.

\bibitem[\protect\citeauthoryear{Erev and Barron}{Erev and Barron}{2005}]{EB05}
Erev, I. and G.~Barron (2005).
\newblock On adaptation, maximization, and reinforcement learning among
  cognitive strategies.
\newblock {\em Psychology Review\/}~{\em 112\/}(4), 912--931.

\bibitem[\protect\citeauthoryear{Feng, Stone, and Tambe}{Feng
  et~al.}{2015}]{FST15}
Feng, F., P.~Stone, and M.~Tambe (2015).
\newblock When security games go green: {D}esigning defender strategies to
  prevent poaching and illegal fishing.
\newblock In {\em Proc.~24th International Joint Conference on Artificial
  Intelligence (IJCAI 2015)}, pp.\  2585--2595.

\bibitem[\protect\citeauthoryear{Ferguson}{Ferguson}{2020}]{Ferguson20}
Ferguson, T.~S. (2020).
\newblock {\em A Course in Game Theory}.
\newblock Singapore: World Scientific.

\bibitem[\protect\citeauthoryear{Halpern, Pass, and Seeman}{Halpern
  et~al.}{2012}]{HPS12}
Halpern, J.~Y., R.~Pass, and L.~Seeman (2012).
\newblock I'm doing as well as {I} can: modeling people as rational finite
  automata.
\newblock In {\em Proc.~Twenty-Sixth National Conference on Artificial
  Intelligence (AAAI '12)}, pp.\  1917--1923.

\bibitem[\protect\citeauthoryear{Kar, Fang, {Delle Fave}, Sintov, and
  Tambe}{Kar et~al.}{2015}]{Kar15}
Kar, D., F.~Fang, F.~{Delle Fave}, N.~Sintov, and M.~Tambe (2015).
\newblock ``{G}ame of {T}hrones'': when human behavior models compete in
  repeated {S}tackelberg security games.
\newblock In {\em Proc. 2015 International Conference on Autonomous Agents and
  Multiagent Systems (AAAMAS '15)}, pp.\  1381--1390.

\bibitem[\protect\citeauthoryear{Kar, Ford, Gholami, Fang, Plumptre, Tambe,
  Driciru, F., Rwetsiba, Nsubaga, and Mabonga}{Kar et~al.}{2017}]{Kar17}
Kar, D., B.~Ford, S.~Gholami, F.~Fang, A.~Plumptre, M.~Tambe, M.~Driciru,
  W.~F., A.~Rwetsiba, M.~Nsubaga, and J.~Mabonga (2017).
\newblock Cloudy with a chance of poaching: {A}dversary behavior modeling and
  forecasting with real-world poaching data.
\newblock In {\em Proc. 16th International Conference on Autonomous Agents and
  Multiagent Systems (AAMAS '16)}, pp.\  159--167.

\bibitem[\protect\citeauthoryear{Lieder, Griffiths, and Hsu}{Lieder
  et~al.}{2018}]{LGH18}
Lieder, F., T.~L. Griffiths, and M.~Hsu (2018).
\newblock Overrepresentation of extreme events in decision making reflects
  rational use of cognitive resources.
\newblock {\em Psychology Review\/}~{\em 125\/}(1), 1--32.

\bibitem[\protect\citeauthoryear{Liu and Halpern}{Liu and
  Halpern}{2020}]{Liu2020MAB}
Liu, X. and J.~Y. Halpern (2020).
\newblock Bounded rationality in {L}as {V}egas: {P}robabilistic finite automata
  play multi-armed bandits.
\newblock In {\em Proc.~36th Conference on Uncertainty in Artificial
  Intelligence (UAI 2020)}, pp.\  1298--1307.
\newblock The proceedings are published as \emph{Proceedings of Machine
  Learning Research}, Vol. 124.

\bibitem[\protect\citeauthoryear{Neyman}{Neyman}{1985}]{Ney85}
Neyman, A. (1985).
\newblock Bounded complexity justifies cooperation in finitely repeated
  prisoner's dilemma.
\newblock {\em Economic Letters\/}~{\em 19}, 227--229.

\bibitem[\protect\citeauthoryear{Nguyen, Yang, Azaria, Kraus, and Tambe}{Nguyen
  et~al.}{2013}]{Nguyen2013SUQR}
Nguyen, T.~H., R.~Yang, A.~Azaria, S.~Kraus, and M.~Tambe (2013).
\newblock Analyzing the effectiveness of adversary modeling in security games.
\newblock In {\em Proc.~Twenty-Seventh National Conference on Artificial
  Intelligence (AAAI '13)}, pp.\  718–724.

\bibitem[\protect\citeauthoryear{Robinson}{Robinson}{1951}]{Robinson51}
Robinson, J. (1951).
\newblock An iterative method of solving a game.
\newblock {\em Annals of Mathematics\/}~{\em 54}, 296--301.

\bibitem[\protect\citeauthoryear{Rubinstein}{Rubinstein}{1986}]{Rub85}
Rubinstein, A. (1986).
\newblock Finite automata play the repeated prisoner's dilemma.
\newblock {\em Journal of Economic Theory\/}~{\em 39}, 83--96.

\bibitem[\protect\citeauthoryear{Stahl}{Stahl}{1993}]{dale1993levelk}
Stahl, D. (1993).
\newblock Evolution of smart$_n$ players.
\newblock {\em Games and Economic Behavior\/}~{\em 5\/}(4), 604--617.

\bibitem[\protect\citeauthoryear{Tambe}{Tambe}{2012}]{Tambe12}
Tambe, M. (2012).
\newblock {\em Security and Game Theory}.
\newblock Cambridge/New York: Cambridge University Press.

\bibitem[\protect\citeauthoryear{Wilson}{Wilson}{2015}]{W02}
Wilson, A. (2015).
\newblock Bounded memory and biases in information processing.
\newblock {\em Econometrica\/}~{\em 82\/}(6), 2257--2294.

\bibitem[\protect\citeauthoryear{Wright and Leyton-Brown}{Wright and
  Leyton-Brown}{2010}]{wright2010beyond}
Wright, J.~R. and K.~Leyton-Brown (2010).
\newblock Beyond equilibrium: {P}redicting human behavior in normal-form games.
\newblock In {\em Proc.~Twenty-Fourth National Conference on Artificial
  Intelligence (AAAI '10)}, pp.\  901--907.

\bibitem[\protect\citeauthoryear{Wright and Leyton-Brown}{Wright and
  Leyton-Brown}{2017}]{WLB17}
Wright, J.~R. and K.~Leyton-Brown (2017).
\newblock Predicting human behavior in unrepeated, simultaneous-move games.
\newblock {\em Games and Economic Behavior\/}~{\em 106}, 16--37.

\bibitem[\protect\citeauthoryear{Wright and Leyton-Brown}{Wright and
  Leyton-Brown}{2019}]{WLB19}
Wright, J.~R. and K.~Leyton-Brown (2019).
\newblock Level-0 models for predicting human behavior in normal-form games.
\newblock {\em Journal of A.I. Research\/}~{\em 64}, 357--383.

\bibitem[\protect\citeauthoryear{Xu, Perrault, Plumptre, Driciru, Wanyama,
  Rwetsiba, and Tambe}{Xu et~al.}{2020}]{XP20}
Xu, L., A.~Perrault, A.~Plumptre, M.~Driciru, F.~Wanyama, A.~Rwetsiba, and
  M.~Tambe (2020).
\newblock Game theory on the ground: {T}he effect of increased patrols on
  deterring poachers.
\newblock In {\em Workshop on AI for Social Good}.
\newblock https://arxiv.org/abs/2006.12411.

\bibitem[\protect\citeauthoryear{Yang, Ford, Tambe, and Lemieux}{Yang
  et~al.}{2014}]{yang2014PAWS}
Yang, R., B.~Ford, M.~Tambe, and A.~Lemieux (2014).
\newblock Adaptive resource allocation for wildlife protection against illegal
  poachers.
\newblock In {\em Proc.~Thirteenth International Joint Conference on Autonomous
  Agents and Multiagent Systems (AAMAS 2014)}, pp.\  453--460.

\end{thebibliography}


\end{document}